\pdfoutput=1
\RequirePackage{ifpdf}
\ifpdf % We~are running pdfTeX in pdf mode
\documentclass[pdftex]{sigma}
\else
\documentclass{sigma}
\fi

\usepackage{mathrsfs,bm}

\newtheorem{prop}{Proposition}[section]
\numberwithin{equation}{section}

\begin{document}
\allowdisplaybreaks

\newcommand{\arXivNumber}{1908.06641}

\renewcommand{\PaperNumber}{004}

\FirstPageHeading

\ShortArticleName{The Arithmetic Geometry of AdS$_2$ and its Continuum Limit}

\ArticleName{The Arithmetic Geometry of AdS$\boldsymbol{{}_2}$\\ and its Continuum Limit}

\Author{Minos AXENIDES~$^\dag$, Emmanuel FLORATOS~$^{\dag\ddag}$ and Stam NICOLIS~$^{\S}$}

\AuthorNameForHeading{M.~Axenides, E.~Floratos and S.~Nicolis}

\Address{$^\dag$~Institute of Nuclear and Particle Physics, NCSR ``Demokritos'',\\
\hphantom{$^\dag$}~Aghia Paraskevi, GR--15310, Greece}
\EmailD{\href{mailto:axenides@inp.demokritos.gr}{axenides@inp.demokritos.gr}}

\Address{$^\ddag$~Physics Department, University of Athens, Zografou University Campus,\\
\hphantom{$^\ddag$}~Athens, GR-15771, Greece}
\EmailD{\href{mailto:mflorato@phys.uoa.gr}{mflorato@phys.uoa.gr}}

\Address{$^\S$~Institut Denis Poisson, Universit\'e de Tours, Universit\'e d'Orl\'eans, CNRS (UMR7013),\\
\hphantom{$^\S$}~Parc Grandmont, 37200 Tours, France}
\EmailD{\href{mailto:stam.nicolis@lmpt.univ-tours.fr}{stam.nicolis@lmpt.univ-tours.fr}}
\URLaddressD{\url{http://www.lmpt.univ-tours.fr/~nicolis/}}

\ArticleDates{Received April 02, 2020, in final form January 02, 2021; Published online January 09, 2021}

\Abstract{According to the 't~Hooft--Susskind holography, the black hole entropy, $S_\mathrm{BH}$, is carried by the chaotic microscopic degrees of freedom, which live in the near horizon region and have a Hilbert space of states of finite dimension $d=\exp(S_\mathrm{BH})$. In previous work we have proposed that the near horizon geometry, when the microscopic degrees of freedom can be resolved, can be described by the AdS$_2[\mathbb{Z}_N]$ discrete, finite and random geometry, where $N\propto S_\mathrm{BH}$.
It has been constructed by purely arithmetic and group theoretical methods and was studied as a toy model for describing the dynamics of single particle probes of the near horizon region of 4d extremal black holes, as well as to explain, in a direct way, the finiteness of the entropy, $S_\mathrm{BH}$. What has been left as an open problem is how the smooth AdS$_2$ geometry can be recovered, in the limit when $N\to\infty$. In the present article we solve this problem, by showing that the discrete and finite AdS$_2[\mathbb{Z}_N]$ geometry can be embedded in a family of finite geometries, AdS$_2^M[\mathbb{Z}_N]$, where $M$ is another integer. This family can be constructed by an appropriate toroidal compactification and discretization of the ambient $(2+1)$-dimensional Minkowski space-time. In this construction $N$ and $M$ can be understood as ``infrared'' and ``ultraviolet'' cutoffs respectively. The above construction enables us to obtain the continuum limit of the AdS$_2^M[\mathbb{Z}_N]$ discrete and finite geometry, by taking both $N$ and $M$ to infinity in a specific correlated way, following a reverse process: Firstly, we show how it is possible to recover the continuous, toroidally compactified, AdS$_2[\mathbb{Z}_N]$ geometry by removing the ultraviolet cutoff; secondly, we show how it is possible to remove the infrared cutoff in a specific decompactification limit, while keeping the radius of AdS$_2$ finite. It is in this way that we recover the standard non-compact AdS$_2$ continuum space-time. This method can be applied directly to higher-dimensional AdS spacetimes.}

\Keywords{arithmetic geometry of AdS$_2$; continuum limit of finite geometries; Fibonacci sequences}

\Classification{14L35; 11D45; 83C57}

\section{Introduction}\label{intro}
\looseness=-1 The present work, mathematically, belongs to the area of algebraic geometry over finite rings. However its relevance for physics stems from the proposal of using specific,
discrete and finite arithmetic geometries, as toy models, in order to describe properties of quantum gravity in general and the structure of space-time, in particular, at distances of the order of the Planck scale ($10^{-33}~\mathrm{cm}$), where the notions of the metric and of the continuity of spacetime break down~\cite{Axenides:2013iwa}.

So it is useful to provide some context, for our study, by presenting a short review of the relevant physics questions about spacetime and quantum gravity. The reader, who is interested only in the mathematical issues of our paper, may skip the following subsection on its physical motivation and resume reading from the subsequent subsection on the outline of the paper.

\subsection{Physics motivation}\label{physicsmot}

At Planck scale energies, quantum mechanics, as we know it from lower energy scales, implies that the notion of spacetime itself becomes ill-defined, through the appearance from the vacuum of real or virtual black holes of Planck length size~\cite{Hawking:1979zw}.

Probing this scale by scattering experiments of any sort of particle-like objects, black holes will be produced and the strength of the gravitational interaction will be of ${\rm O}(1)$, which leads to a breakdown of perturbative gravity and of the usual continuum spacetime description~\cite{Carlip:2009km,Carlip:2011tt}.

 The above remarks led some authors to consider the idea, that one has to abandon continuity of spacetime, locality of interactions and regularity of dynamics.
 Indeed there are recent arguments that quantization of gravity implies discretization and finiteness of space time~\cite{Hooft:2016pmw,Verlinde:2010hp}. This is, indeed, an old idea, that was put forward, already way back, by the founders of quantum physics and gravity.

The most successful and popular framework today, to tackle this fundamental problem, is considered to be the AdS/CFT correspondence. It attempts to define spacetime geometry~-- and thus gravity~-- as an emergent phenomenon, that must and can be described in the language of conformal field theory. The realization of this correspondence has passed many non-trivial consistency checks by explicit calculations, that are valid for distances of the order of $R_{{\rm AdS}_2}$, much larger than the Planck scale in the bulk space-time.

An example of such a non-trivial check consists in providing, on the one hand, the degrees of freedom that can account for the black hole entropy~\cite{Sen:2009vz, Sen:2008vm} and recovering the Bekenstein--Hawking entropy, at length scales much larger than the Planck length, along with a certain class of corrections; on the other hand, in providing a resolution and reformulation of the so-called ``old black hole information paradox''~\cite{Hawking:2005kf,Maldacena:2020ady}.

When the curvature of the bulk spacetime becomes locally of the order of the Planck scale, the holographically dual conformal field theory on the boundary, becomes a free field theory~-- but the complexity of the problem of understanding the space-time geometry and gravity of the bulk appears in the guise of the construction of the infinitely ``complicated'' operators of the free boundary conformal field theory. This is necessary for representing ``local events'' in the bulk, as well as the bulk diffeomorphisms (the so-called ``problem of locality'' in the AdS/CFT correspondence, presented, for instance in~\cite{Almheiri:2014lwa}). This phenomenon is a consequence of the so-called UV/IR correspondence, that is inherent in the AdS/CFT framework. How to resolve it is, at present, under study.

A few years ago the seminal paper~\cite{Almheiri:2012rt} highlighted the relevance of the so-called ``new black hole information paradox''~\cite{Papadodimas:2012aq},
 which finally lead to the conjectures that go under the label ER~=~EPR~\cite{Maldacena:2013xja} and culminate in the so-called QM~=~GR correspondence~\cite{Susskind:2017ney}.

These conjectures relate strongly the description of spacetime geometry and quantum gravity to quantum information theoretic tools, such as entanglement of information, algorithmic complexity,random quantum networks,quantum holography, error correcting codes etc.

On the other hand a direct discussion of the structure of the bulk spacetime at Planck scale lies beyond the present capabilities of the holographic AdS/CFT framework.

We return, in what follows, to the ideas of discretization of spacetime which constitutes the framework of our study.

The hypothesis for a discrete and finite spacetime for quantum gravity is a possible way for explaining the remarkable fact that the Hilbert space of states of the BH microscopic degrees of freedom is finite-dimensional. Its dimensionality equals to the exponential of the Bekenstein--Christodoulou--Hawking black hole entropy, which is of quantum origin. The generalization of the Bekenstein entropy bounds implies that, for any pair of local observers in a general gravitational background, the physics inside their causal diamond is also described by a finite-dimensional Hilbert space of states~\cite{Bousso:2018bli}. This result has been exploited further and consistently under the name of holographic spacetime, in the works~\cite{Banks:2020dus,Bao:2017rnv, Giddings:2012bm}.

Our idea about the nature of spacetime at the Planck scale, takes the notion of a holographic spacetime one step further: Namely, that the finite dimensionality of the Hilbert space of local spacetime regions originates from a discrete and finite spacetime, which underlies the emergent continuous geometric description~\cite{Axenides:2013iwa, Floratos:1989au}.

Our starting point is essentially the hypothesis that spacetime, at the Planck scale, is fundamentally discrete and finite
and, moreover, does not emerge from any other continuous description(conformal field theory, string theory, or anything else). We claim that, at ``large'' distances (in units of the Planck length), the continuous spacetime geometry can be described as an infrared limit thereof. This hypothesis, indeed, is similar to the proposal by 't Hooft~\cite{Hooft:2016pmw}.

This assumption then entails developing and using the appropriate mathematical tools, that can describe the properties and dynamics of discrete and finite geometries as well as the emergence, in their infrared limits, of continuous geometries.

We do not wish to imply that it is not possible to define quantum gravity, with a finite-dimensional Hilbert space, in any other way; just that this is one possible way to describe quantum spacetime with a finite-dimensional Hilbert space.

We shall now present a short review of our recent work, along with the outline of our paper.

\subsection{Context and outline of the paper}
In our previous work we have proposed a discrete and finite model geometry, which we have called AdS$_2[\mathbb{Z}_N]$, for any, positive, integer $N$. This geometry is simply defined as the set of points of integer entries, $(k,l,m)$, that satisfy the relation $k^2+l^2-m^2\equiv 1\,\mathrm{mod}\,N$. Thus, we have replaced, in the definition of the continuous AdS$_2$ geometry, the real numbers with the finite ring of integers mod $N$.

AdS$_2[\mathbb{Z}_N]$, defined in this way, has a random structure, due to the modular arithmetic. It is well known that deterministic processes, using modular arithmetic, can produce deterministic random sequences of points~\cite{cesaratto2008non, knuth1981art}.

As explained in~\cite{Axenides:2013iwa, Axenides:2015aha,Axenides:2016nmf}, this particular discretization is chosen, among many possible ones, because it supports the holographic correspondence between the bulk, AdS$_2[\mathbb{Z}_N]$, and its boundary, $\mathbb{P}^1[\mathbb{Z}_N]$, the discrete projective line. The reason this discrete holography exists at all is that it is possible to realize the action of its discrete and finite isometry group, ${\rm PSL}_2[\mathbb{Z}_N]$, in two ways: Firstly, as the isometry group of the bulk and secondly, as the (Möbius) conformal group on the boundary.

In this approach a long-standing question has been the meaning and existence of a continuum limit of the finite and random modular geometry as $N\to\infty$; of equal importance is, of course whether the usual smooth AdS$_2$ geometry can be recovered in this way at all.

In the present paper we will demonstrate firstly that this limit exists and, secondly that, in fact, the continuous geometry of AdS$_2$ emerges from the discrete
AdS$_2^M[\mathbb{Z}_N]$ geometry as an infrared limit. In order to show this we reconstruct this geometry from AdS$_2$, in two steps:
\begin{enumerate}\itemsep=0pt
\item
The first step involves the discretization of AdS$_2$, using an appropriate spacetime lattice in the ambient $(2+1)$-dimensional, Minkowski, spacetime. This requires introducing an ultraviolet cutoff $a=R_{{\rm AdS}_2}/M$, for any integer~$M$.
The lattice spacing $a$ has the important property that it breaks the continuous Lorentz group to its arithmetic discrete subgroup ${\rm SO}(2,1,\mathbb{Z})$~\cite{Schild48, Schild49}. The Minkowski spacetime lattice induces moreover, on the continuum~AdS$_2$, an infinite set of integral points with isometry group ${\rm SO}(2,1,\mathbb{Z})$. This set defines an integral lattice of~AdS$_2$, for any~$M$, which we shall call henceforth AdS$_2^M[\mathbb{Z}]$.

The reason for introducing a sequence of lattices is that it allows us to define
the continuum limit by taking $a\to 0$ or $M\to\infty$, keeping $R_{{\rm AdS}_2}$ fixed.

In this limit the sequence of spaces AdS$_2^M[\mathbb{Z}]$ tends to AdS$_2$ in the topology of the ambient, flat, three-dimensional spacetime.
\item
The second step involves the introduction of an infrared cutoff, $L=a N$, where $N > M$, such that, in the limit $M\to\infty$ and $N\to\infty$, the ratio $N/M$ tends to a finite value, $L/R_{{\rm AdS}_2}\equiv\gamma > 1$. It should be noted that this construction remains consistent, also when $N< M$ and $\gamma < 1$. The difference is that in this case the throat doesn't lie within the enclosing box.

The introduction of the infrared cutoff is realized by symmetrically enclosing a region of the throat of the AdS$_2$ hyperboloid, as large as desired, in a box, of size $L$, in the ambient spacetime and then imposing periodic boundary conditions. In this way, we obtain the AdS$_2$ hyperboloid infinitely folded due to the periodic boundary conditions. It is possible to recover the unfolded AdS$_2$ complete geometry by removing the infrared cutoff in the limit $L\to\infty$.

 On the other hand, the introduction of the periodic box of size $L=Na$
identifies all the points of the integral lattice, whose coordinates differ by integer multiples of $N$.

This equivalence relation implies that all the points of AdS$_2^M[\mathbb{Z}]$ can be classified by a~finite number of equivalence classes represented by points
{\em inside} the box. However, these representatives need not lie on the part of AdS$_2^M[\mathbb{Z}]$ that's enclosed by this box. We observe in addition that the IR cutoff, $N$, deforms the ${\rm SO}(2,1,\mathbb{Z})$ symmetry of the integral lattice to its $\mathrm{mod}\,N$ reduction, ${\rm SO}(2,1,\mathbb{Z}_N)$.

The coordinates of all points $(k,l,m)$, of AdS$_2^M[\mathbb{Z}]$, that satisfy the equation
\begin{equation*}%\label{AdS2MZNdef}
k^2+l^2-m^2\equiv M^2\,\mathrm{mod}\,N
\end{equation*}
define the finite geometry $\mathrm{AdS}_2^M[\mathbb{Z}_N]$.

In order to identify the solutions of the above equation with the elements of AdS$_2[\mathbb{Z}_N]$, it is necessary to impose that
\begin{equation*}%\label{UVIRmixing}
M^2\equiv 1\,\mathrm{mod}\,N.
\end{equation*}
 This condition provides a relation between the UV and IR cutoffs, $M$ and $N$.
\end{enumerate}
Having reconstructed the finite geometry $\mathrm{AdS}_2^M[\mathbb{Z}_N]$, by the two-step process, discretization and toroidal compactification, we are able to show that the continuum limit can be taken by finding infinite sequences of UV/IR cutoff pairs $\{(M_n,N_n)\}$, under the constraint $M^2\equiv 1\,\mathrm{mod}\,N$,that satisfy the conditions described in the two-step process, mentioned above.

The main result of our paper is the explicit construction of the continuum limit, by a reverse, two-step, process:
\begin{enumerate}\itemsep=0pt
\item
First, we remove the UV cutoff, using pairs of UV/IR cutoffs, chosen from the $k$-Fibonacci sequences, which lead to different values of the ratio $\gamma$, for different values of $k$.
\item
Next, we remark that the ratio, $L/R_{{\rm AdS}_2}$ and, thus, the
IR cutoff, $L$, is an increasing function of $k$ and, therefore, in the large $k$ limit, we can remove the IR cutoff, while keeping the radius, $R_{{\rm AdS}_2}$ fixed, but arbitrary.
\end{enumerate}

The plan of our paper is as follows:

Section~\ref{NHEG} consists of two subsections: In Section~\ref{holo} we recall the salient features of the geometry of AdS$_2$ as a ruling surface and as a coset space.
In Section~\ref{modgeom}, we describe the coset structure and the ruling property of the finite geometry, AdS$_2[\mathbb{Z}_N]$ and we discuss the problem of counting its points, for~$N$ a power of a~prime integer. Using the Chinese remainder theorem, we find the number of points, for any integer~$N$. We find that the ruling property leads to a~consistent description with one chart, when $N=p^r$ and $p\,\mathrm{mod}\, 4\equiv 3$ and requires two charts, if $p\,\mathrm{mod}\,4\equiv 1$.

Section~\ref{modN} consists of two subsections: In Section~\ref{UVcutofflatt} we introduce a lattice in the embedding two time-one space Minkowski spacetime, $\mathscr{M}^{2,1}$, with lattice spacing $a=R_{{\rm AdS}_2}/M$. This ``ultraviolet (UV) cutoff'' induces the integral lattice of points of AdS$_2$, which we call AdS$_2^M[\mathbb{Z}]$.
We identify the isometry group of AdS$_2^M[\mathbb{Z}]$ as ${\rm SO}(2,1,\mathbb{Z})$ for any positive integer $M$ and present a review of its basic features. Moreover we show that all the points of AdS$_2^M[\mathbb{Z}]$ lie on light-like lines, which intersect the circle of the throat at rational points. Furthermore, on each light-like line, there is an infinite number of, randomly distributed, integral points.

In Section~\ref{UVIR} we compactify the embedding Minkowski spacetime, $\mathscr{M}^{2,1}$, inside a torus,~$\mathbb{T}^3$, of size $L=Na$, where $N$ is an integer, larger than~$M$, by imposing periodic boundary conditions. This is equivalent to identifying the points, whose coordinates differ by integral multiples of~$L$.
The continuum AdS$_2$, after such a compactification, becomes infinitely folded inside the torus. The infinite number of points of AdS$_2^M[\mathbb{Z}]$ gets mapped to a set of a finite number of points, which defines a finite geometry, AdS$_2^M[\mathbb{Z}_N]$. The isometry group of this geometry is found to be ${\rm SO}(2,1,\mathbb{Z}_N)$ for all $M$, which is the reduction mod $N$ of the group ${\rm SO}(2,1,\mathbb{Z})$.

In Section~\ref{contlim} we construct the continuum limit of AdS$_2[\mathbb{Z}_N]$, by relating it to the geometry AdS$_2^M[\mathbb{Z}_N]$. This is achieved by imposing the constraint $M^2\equiv 1\,\mathrm{mod}\,N$. In Section~\ref{fibon} we construct a sequence of UV/IR pairs, $(M_n,N_n)$, $n=1,2,3,\ldots$, that belong to the Fibonacci sequence $f_n$, with the properties mentioned previously.
The limit $n\to\infty$ corresponds to the continuum limit $a\to 0$, where the UV cutoff, with {\em fixed} IR cutoff $L=R_{{\rm AdS}_2}\gamma$, has been removed.
In Section~\ref{contfrac} we show how the IR cutoff can be removed, once we consider UV/IR pairs belonging to the $k$-Fibonacci sequences, by taking the limit $k\to\infty$.

In Section~\ref{concl} we draw our conclusions and present our ideas for further inquiry.

\section[\protect{Continuum AdS2 and the AdS2[ZN] modular geometry}]{Continuum $\boldsymbol{{\rm AdS}_2}$ and the $\boldsymbol{{\rm AdS}_2[\mathbb{Z}_N]}$ modular geometry}\label{NHEG}
\subsection[AdS2 geometry as a ruling surface and as a coset space]{$\boldsymbol{{\rm AdS}_2}$ geometry as a ruling surface and as a coset space}\label{holo}

In the near horizon region of spherically symmetric 4d extremal black holes the geometry is known to be of the form
AdS$_2\times S^2$, where the AdS$_2 =\mathrm{PSL}(2,\mathbb{R})/{\rm PSO}(1,1,\mathbb{R})$ factor describes the geometry of the radial and time coordinates
 and $S^2$ is the horizon surface.

 In the present work we will develop the necessary mathematical framework which will enable us to discretize consistently the ${\rm AdS}_2$ factor, leaving for a future publication the discretization of the $S^2$ factor.

Indeed, we shall review the salient features of the continuum AdS$_2$, geometry as a single-sheeted 2d hyperboloid, considered both as a ruled surface and as a coset space~\cite{BengtssonAdS,Gibbons:2011sg}. Both of these descriptions are amenable to consistent discretization as we shall see in the following sections.

The AdS$_2$ spacetime is a one-sheeted hyperboloid, defined through its
global embedding in Minkowski spacetime with one space~-- and two time-like
dimensions by the equation~\cite{Cadoni:1999ja,Patricot:2004kc}.{\samepage
\begin{equation*}%\label{AdS2_M21}
x_0^2 + x_1^2 - x_2^2 = R_{{\rm AdS}_2}^2.
\end{equation*}
We shall work in units where $R_{{\rm AdS}_2}=1$.}

The boundaries of AdS$_2$ consist of two time-like disconnected circles, where
AdS$_2$ approaches, asymptotically, the light cone of ${\mathscr M}^{2,1}$,
\begin{equation*}%\label{M21_LC}
x_0^2 + x_1^2 - x_2^2 = 0.
\end{equation*}

AdS$_2$ can be, also, described as the homogeneous space ${\rm SO}(2,1)/{\rm SO}(1,1)$. This case is special, in that ${\rm SO}(2,1)$ has a double cover, ${\rm SL}(2,\mathbb{R})$, so that we have $\mathrm{AdS}_2 = \mathrm{PSL}(2,\mathbb{R})/{\rm PSO}(1,1)$.

In order to establish our notation and conventions, we proceed with the Weyl construction of the
double covering group, ${\rm PSL}(2,\mathbb{R})$.

To every point, $x_\mu\in\mathrm{AdS}_2$, $\mu=0,1,2$, we assign the
traceless, real, $2\times 2$ matrix
\begin{equation*}%\label{Weyl}
{\sf M}(x)\equiv\left(\begin{matrix} x_0 & x_1+x_2 \\x_1-x_2 & -x_0\end{matrix}\right).
\end{equation*}
Its determinant is $\det {\sf M}(x)=-x_0^2-x_1^2+x_2^2=-1$.

The action of any element ${\sf A}$ of the isometry group SL$(2,\mathbb{R})$ on AdS$_2$ is defined through the
mapping
\begin{equation*}%\label{Weyl_mapping}
{\sf M}(x') = {\sf A}{\sf M}(x){\sf A}^{-1}.
\end{equation*}

This induces an ${\rm SO}(2,1)$ transformation on $(x_0,x_+,x_-)$, where $x_\pm=x_1\pm x_2$,
\begin{equation*}%\label{induced_transf}
x' \equiv \Lambda({\sf A}) x.
\end{equation*}
More concretely, when
\begin{equation*}%\label{AinSL2R}
{\sf A} = \left(\begin{matrix} a & b \\ c & d\end{matrix}\right)
\end{equation*}
 the induced Lorentz transformation, $\Lambda({\sf A})$, in the light cone basis $(x_0,x_+,x_-)$, is given by the expression
\begin{equation*}%\label{LofA}
\Lambda({\sf A})=\left(\begin{matrix}ad+bc & -ac & bd \\
 -2ab & a^2 & -b^2\\
 2cd & -c^2 & d^2\end{matrix}\right).
\end{equation*}

Choosing as the origin of coordinates the base point $\bm{p}\equiv (1,0,0)$, its
stability group ${\rm SO}(1,1)$ is the group of Lorentz transformations in the
$x_0=0$ plane of ${\mathscr M}^{2,1}$ or, equivalently, the ``scaling''
subgroup, ${\sf D}$, of ${\rm SL}(2,\mathbb{R})$
\begin{equation*}%\label{scaling}
{\sf D}\ni {\sf S}(\lambda)\equiv \left(\begin{matrix} \lambda & 0 \\ 0 & \lambda^{-1}\end{matrix}\right)
\end{equation*}
for $\lambda\in\mathbb{R}^\ast$.

For this choice of the stability point, we define the coset, $h_{\sf A}$, by decomposing ${\sf A}$ as
\begin{equation*}%\label{Acoset}
{\sf A} = h_{\sf A}{\sf S}(\lambda_{\sf A}).
\end{equation*}
Thus, we associate uniquely to every point $x\in\mathrm{AdS}_2$ the corresponding coset representati\-ve~$h_{\sf A}(x)$.

We herein introduce the global coordinate system, defined by the straight lines that generate AdS$_2$ and for which it can be checked easily that they form its complete set of light cones.

Consider the two lines, $\bm{l}_\pm(\bm{p})$, passing through the point $\bm{p}\in{\mathscr M}^{2,1}$,
orthogonal to the $x_0$ axis and at angles $\pm\pi/4$ to the $x_1=0$ plane. They are defined by the intersection of AdS$_2$ and the plane $x_0=1$, cf.\ Fig.~\ref{LCAdS2}.

The coordinates of any point, $\bm{q}_+\in\bm{l}_+(\bm{p})$,
 $\bm{q}_-\in\bm{l}_-(\bm{p})$ are given as $(1,\mu_\pm,\pm\mu_\pm)$,
 $\mu_\pm\in\mathbb{R}$ correspondingly.

 We can parametrize any point $x_\mu$, of AdS$_2$, by the intersection of the local light cone lines,~$\bm{l}_\pm(x)$, with coordinates $\mu_\pm$ and $\phi_\pm$ through the relations
\begin{gather*}
x_0 = \cos\phi_\pm -\mu_\pm\sin\phi_\pm,\qquad
x_1 = \sin\phi_\pm +\mu_\pm\cos\phi_\pm,\qquad
x_2 = \pm\mu_\pm. %\label{new_points}
\end{gather*}
These can be inverted as follows
\begin{equation}
\label{inverse_mapping}
{\rm e}^{\mathrm{i}\phi_\pm} = \frac{x_0 + \mathrm{i}x_1}{1\pm\mathrm{i} x_2},\qquad
\mu_\pm = \pm x_2.
\end{equation}
The geometric meaning of the coordinates $\phi$ and $\mu$ is that $\mu$ parametrizes the $x_2$, space-like, coordinate and, thus, $\mu_\pm\sqrt{2}$ parametrizes the light cone lines $\bm{l}_\pm(x)$. The angle $\phi_\pm$ is the azimuthal angle of the intersection of $\bm{l}_\pm(x)$ with the plane $(x_0,x_1)$. From equation~(\ref{inverse_mapping}), by re-expressing numerator and denominator in polar coordinates, we find
\begin{equation}\label{polar_inverse_mapping}
\phi=\tau-\sigma,
\end{equation}
where $\tau$ and $\sigma$ are the arguments of the complex numbers $x_0+\mathrm{i}x_1$ and $1+\mathrm{i}x_2$.

The corresponding coset parametrization (group coset motion which brings the origin to the point $x$) is
\begin{equation*}%\label{cosets}
h(\mu_\pm,\phi_\pm) = {\sf R}(\phi_\pm){\sf T}_\pm(\mu_\pm),
\end{equation*}
where
\begin{equation*}%\label{coset_rot}
{\sf R}(\phi) = \left(\begin{matrix} \cos\phi/2&
 -\sin\phi/2\\\sin\phi/2& \cos\phi/2\end{matrix}\right)
\end{equation*}
and
\begin{equation*}%\label{coset_trans}
{\sf T}_+(\mu) = \left[{\sf T}_-(-\mu)\right]^\mathrm{T} =
\left(\begin{matrix} 1 & -\mu\\ 0 & 1\end{matrix}\right).
\end{equation*}
It is easy to see also, that ${\sf T}_\pm(\mu_\pm)$, acting on the base point
$X(\bm{p})$, generate the light cone~$l_\pm(\bm{p})$, so we identify these one parameter subgroups with the light cones at $\bm{p}$.

\begin{figure}[t]\centering
\includegraphics[scale=0.6]{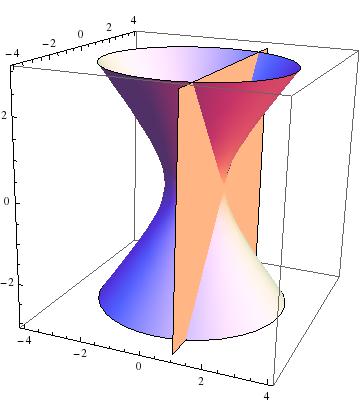}
\caption{The light cone of AdS$_2$ at $\bm{p}=(1,0,0)$.}\label{LCAdS2}
\end{figure}

In the literature the study of field theories on AdS$_2$ requires an extension, to the universal covering of this spacetime, $\widetilde{\mathrm{AdS}_2}$, together with appropriate boundary conditions, in order to avoid closed time-like geodesics and reflection of waves from the boundary. This extension can be parametrized using as time coordinate the azimuthal angle $\tau$, by extending its range, $(-\pi,\pi)$ to $(-\pi +2\pi k,\pi+ 2\pi (k+1))$, $k=\pm1,\pm2,\ldots$ and by the space coordinate $\sigma\in(-\pi/2,\pi/2)$, defined in equations~(\ref{polar_inverse_mapping}). The extension of the range of $\tau$ parametrizes the infinitely-sheeted Riemann surface of the function $\log(\cdot)$, used in deriving equation~(\ref{polar_inverse_mapping}).

It is interesting to note that the coset structure of AdS$_2$ can be elevated to $\widetilde{\mathrm{AdS}_2}$ by using the universal covering group of ${\rm SL}(2,\mathbb{R})$, $\widetilde{{\rm SL}(2,\mathbb{R})}$, which has been explicitly constructed in~\cite{rawnsley}.

\subsection[\protect{The discrete modular geometry AdS2[ZN] and its isometries}]{The discrete modular geometry $\boldsymbol{{\rm AdS}_2[\mathbb{Z}_N}]$ and its isometries}\label{modgeom}
We propose to model the random, non-local geometry of the near horizon region of extremal black holes by a number-theoretic discretization of the AdS$_2$ factor, that preserves its group-theoretical structure.

This is done by replacing the continuous coset structure of AdS$_2$, presented in the previous section, by the discrete cosets,
\begin{equation*}%\label{discrete_cosets}
\mathrm{AdS}_2[ \mathbb{Z}_N] =\mathrm{PSL}(2,\mathbb{Z}_N)/\mathrm{PSO}(1,1,\mathbb{Z}_N).
\end{equation*}
We thereby replace the set of real numbers, $\mathbb{R}$, by the set of integers modulo~$N$. We called this a ``modular discretization'' of AdS$_2$ in~\cite{Axenides:2013iwa}.

This is a finite, deterministically random, set of points in the embedding Minkowski spacetime~$\mathscr{M}^{2,1}$.

By introducing appropriate length scales and by taking the large $N$ limit we shall show in the following sections how the smooth geometry of AdS$_2$ can emerge.

We notice some interesting factorizations of the algebraic structures with respect to the integer $N$: If $N=N_1 N_2$, with $N_{1,2}$ coprime, then we have~\cite{Athanasiu:1998cq}
\begin{equation*}%\label{factorizationSL2ZN}
\mathrm{PSL}(2,\mathbb{Z}_{N_1 N_2}) = \mathrm{PSL}(2,\mathbb{Z}_{N_1})\times
 \mathrm{PSL}(2,\mathbb{Z}_{N_2})
\end{equation*}
and
\begin{equation*}%\label{factorizationAdS2N}
\mathrm{AdS}_2[\mathbb{Z}_{N_1 N_2}] = \mathrm{AdS}_2[\mathbb{Z}_{N_1}] \times \mathrm{AdS}_2[\mathbb{Z}_{N_2}].
\end{equation*}
These factorizations imply that all powers of primes, $2^{n_1}, 3^{n_2},5^{n_3},\ldots$, are the building blocks of our construction. The physical interpretation of this factorization is that the most coarse-grained Hilbert spaces on the horizon have dimensions powers of primes.

In order to study the finite geometry of AdS$_2[p^r]$, we recall the following facts about its ``isometry group'' ${\rm PSL}(2,\mathbb{Z}_{p^r})$:
\begin{itemize}\itemsep=0pt
\item
The order of ${\rm PSL}(2,\mathbb{Z}_{p^r})$ is $p^{3r-2}\big(p^2-1\big)/2$~\cite{sp2norder}.
\item
The set of points of the finite geometry of AdS$_2[p^r]$ is, by definition, the set of all solutions of the equation
\begin{equation*}%\label{AdS2_p}
x_0^2 + x_1^2 -x_2^2\equiv 1\,\mathrm{mod}\,p^r.
\end{equation*}
The elements of this set can be parametrized as follows
\begin{equation*}%\label{LCAdS2N}
x_0\equiv (a-b\mu)\,\mathrm{mod}\,p^r, \qquad x_1\equiv (b+a\mu)\,\mathrm{mod}\,p^r, \qquad
x_2\equiv \mu\,\mathrm{mod}\,p^r.
\end{equation*}
where $a^2 + b^2\equiv 1\,\mathrm{mod}\,p^r$ and $a, b, \mu\in\mathbb{Z}_{p^r}$.
\item The points of AdS$_2[p^r]$ comprise the bulk~-- for the holographic correspondence~-- to these we must include the points on the boundary.

The boundary is the ``mod $p^r$'' projective line, $\mathbb{P}^1[\mathbb{Z}_{p^r}]$, defined as the set
\begin{equation*}%\label{RP1p}
\mathbb{P}^1[\mathbb{Z}_{p^r}] = \mathbb{Z}^\ast[p^r]\cup\{0,\infty\},
\end{equation*}
so the number of boundary points (cosets) is $p^{r-1}(p-1)+2$.
\end{itemize}

We shall focus henceforth on the properties of the random set of points, that constitute the bulk, i.e., AdS$_2[N=p^r]$.

 The deterministic randomness of the points of AdS$_2^M[\mathbb{Z}_N]$ can be illuminated from their representation in the three-dimensional ambient space-time, cf.\ Fig.~\ref{IntPointsMod}.
\begin{figure}[t!]\centering
\includegraphics[scale=0.15]{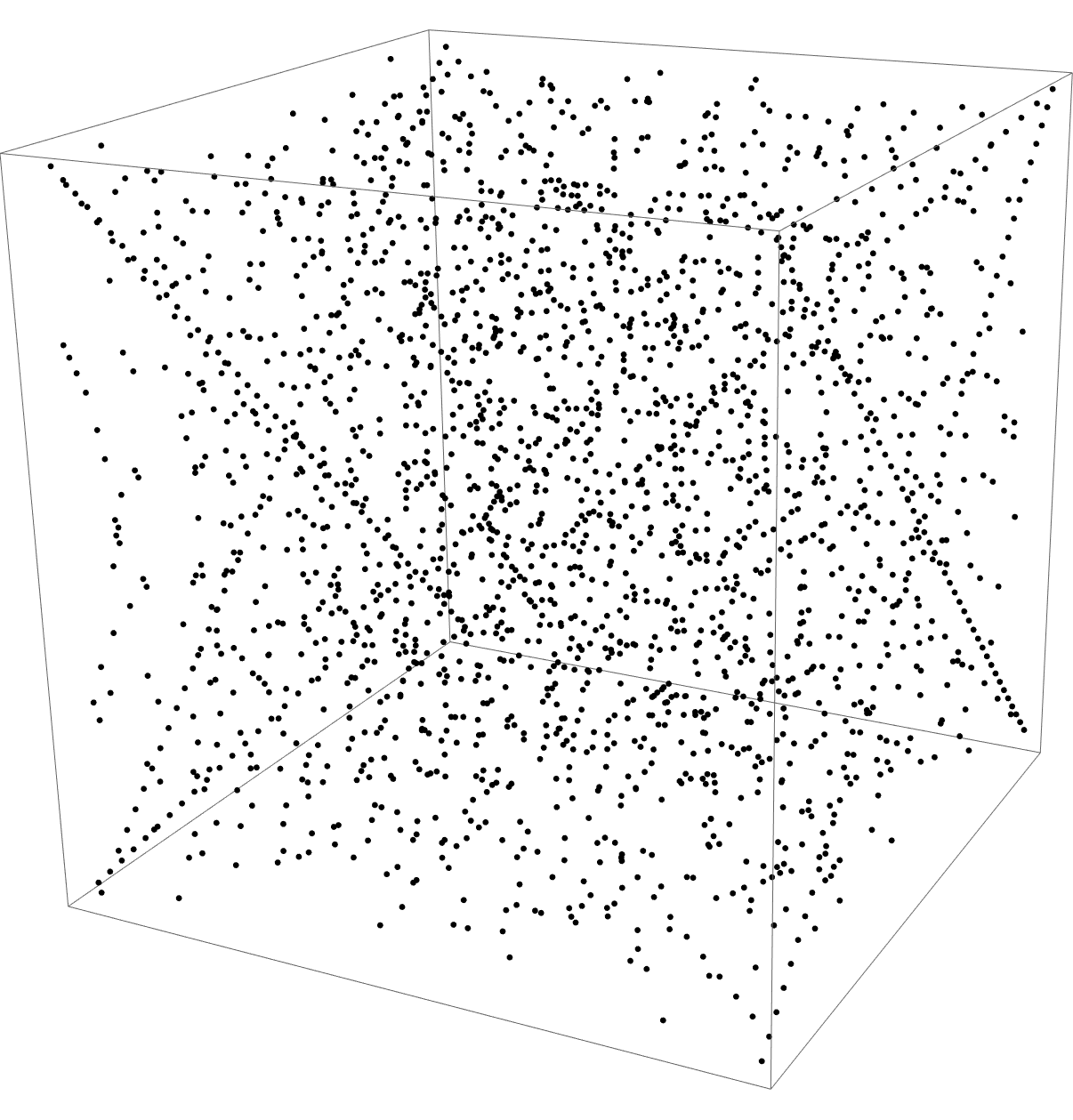}
\caption{The integral points, $(k,l,m)$, that satisfy $k^2+l^2-m^2\equiv 1\,\mathrm{mod}\,47$, i.e., that belong to AdS$_2[47]$.}
\label{IntPointsMod}
\end{figure}

\begin{prop}It is interesting to notice, that, in analogy with the continuous case, it is possible to define, for ${\rm AdS}_2[N]$, a global ruling parametrization for $N=p^r$, where $p$ is a prime of the form $(a)$~$p\equiv 3\,\mathrm{mod}\,4$, while when $(b)$~$p\equiv 1\,\mathrm{mod}\,4$, we need two charts to obtain all such points.
\end{prop}

\begin{proof} We, start, by parametrizing the points of AdS$_2^M[\mathbb{Z}_N]$ by the ruling of the discrete line $\bm{l}=(1,\mu,\mu)$ around the discrete circle of the throat of AdS$_2^M[\mathbb{Z}_N]$:
\begin{equation*}%\label{ads2discrete}
x_0=a-\mu b,\qquad x_1=b+\mu a,\qquad x_2=\mu,
\end{equation*}
where $a,b,\mu\in\mathbb{Z}_N$ and $a^2+b^2\equiv 1\,\mathrm{mod}\,N$, cf.\ Fig.~\ref{discretecirclemod1001}.

\begin{figure}[t!]\centering
\includegraphics[scale=0.3]{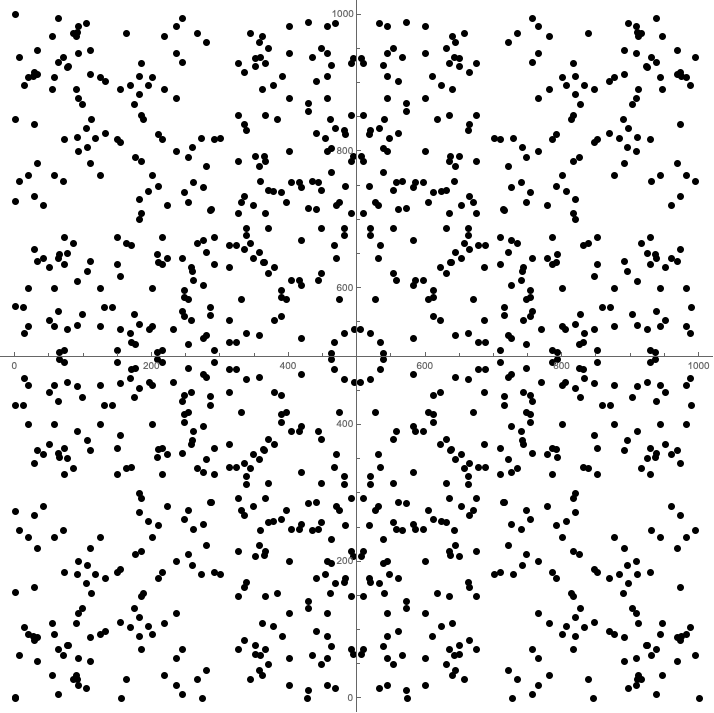}
\caption{The points of the discrete circle, $a^2+b^2\equiv 1\,\mathrm{mod}\,1001$.}\label{discretecirclemod1001}
\end{figure}

This parametrization suffices to generate {\em all} the points, for case~(a), as an explicit comparison with direct counting confirms; for case~(b), we must add a second parametrization, by exchanging~$x_0$ and~$x_1$. The reason this is necessary is that, in case~(b), given $x_0$ and $x_1$ it's not possible to obtain~$a$ and~$b$, since there exists a $\mu=\mu_0$, such that $\mu_0^2\equiv -1\,\mathrm{mod}\,N$ in this case.
\end{proof}

We shall now proceed in counting the integral points of AdS$_2[N]$, for any integer $N$.

\begin{prop} When $N=p^r$, numerical experiments suggest the following recursion relation for the number of points of ${\rm AdS}_2[\mathbb{Z} [p^r]$, ${\sf Sol}(p^r)$,
\begin{equation*}%\label{solpk}
{\sf Sol}(p^r)=p^{2(r-1)}{\sf Sol}(p)\Rightarrow {\sf Sol}(p^r)=p^{2r-1}(p+1),
\end{equation*}
where ${\sf Sol}(p)=p(p+1)$ and $r=1,2,\ldots$ for any prime integer~$p$.
\end{prop}

This proposition can be proved, by using the coset property of AdS$_2[p^r]$.

\begin{proof}
 The rank of the group ${\rm PSL}_2[p^r]$ is known to be $p^{3 r-2} \big(p^2-1\big)/2$, while that of its dilatation subgroup ${\rm PSO}(1,1,p^r)$ is $p^{r-1} (p-1)/2$.
 This is a consequence of the fact that the rank is equal to the number of invertible integers modulo $p^r$ divided by 2 (due to its projective structure). Thus, since AdS$_2[\mathbb{Z}_{p^r}]$, is identified with the coset geometry ${\rm PSL}_2[p^r]/{\rm PSO}(1,1,p^r)$,
we get the promised result, $p^{2 r-1} (p+1)$.
\end{proof}

\begin{figure}[t]\centering
\includegraphics[scale=0.35]{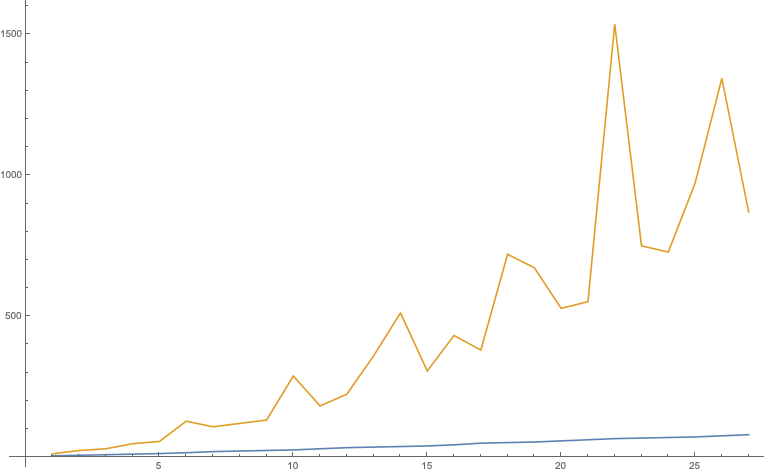}
\caption{The number of solutions to $k^2+l^2-m^2=1$ (blue curve) and $k^2+l^2-m^2\equiv 1\,\mathrm{mod}\,N$ (yellow curve), for $3\leq N\leq 29$ obtained by exact enumeration.}\label{PinakasShmeiwn29}
\end{figure}

The case $N=2^n$ is special: We find ${\sf Sol}(2)=4$, ${\sf Sol}(4)=24$, and ${\sf Sol}\big(2^k\big)=4{\sf Sol}(2^{k-1})$, for $k\geq 3$. We remark that $N=4$ is an exception.
The solution is ${\sf Sol}\big(2^k\big)=2^{2k+1}$, for $k\geq 3$. We plot the results of exact enumeration in Fig.~\ref{PinakasShmeiwn29} for $3\leq N\leq 29$.
We notice that there are peaks for composite values of~$N$. The additional points count the equivalence classes of points of AdS$_2[\mathbb{Z}]$ ${\rm mod}\,N$.

From these results we deduce that, for large $N$, the number of solutions, ${\rm mod}\,N$, scales like the area, i.e.,~$N^2$.

\section[Discretization and toroidal compactification of the AdS2 geometry]{Discretization and toroidal compactification\\ of the $\boldsymbol{{\rm AdS}_2}$ geometry}\label{modN}

\subsection[\protect{The UV cutoff, the lattice of integral points and the SO(2,1,Z) isometry of AdS2M[Z]}]{The UV cutoff, the lattice of integral points\\ and the $\boldsymbol{{\rm SO}(2,1,\mathbb{Z})}$ isometry of $\boldsymbol{{\rm AdS}_2^M[\mathbb{Z}]}$}\label{UVcutofflatt}

 We shall now present and study in detail the lattice of integral points of AdS$_2$, along with its isometries.

The physical lengthscale in our problem is the radius of the AdS$_2$ spacetime, $R_{{\rm AdS}_2}$. We set $R_{{\rm AdS}_2}=1$ and we divide it into $M$ segments, of length $a=R_{{\rm AdS}_2}/M$.
This defines $a$ as the UV cutoff (lattice spacing) and $M\in\mathbb{N}$ and, hence, a lattice in $\mathscr{M}^{2,1}$.

The continuum limit is defined by taking $M\to\infty$ and $a\to 0$ with $R_{{\rm AdS}_2}=1$ fixed.

The global embedding coordinates $(x_0, x_1,x_2)$ of this lattice are $(ka,la,ma)=a(k,l,m)$, where $k,l,m\in\mathbb{Z}$. They are measured in units of the lattice
spacing $a$. Therefore the lattice points, that lie on AdS$_2$ satisfy the equation
\begin{equation*}%\label{AdS2latt}
k^2+l^2-m^2=M^2,
\end{equation*}
whose solutions define AdS$_2^M[\mathbb{Z}]$, the set of all integral points of AdS$_2$ with integer radius~$M$.

In the literature there has been considerable effort in counting the number of solutions to the above equation, in particular the asymptotics of the density of such points~\cite{Baragar,Duke,Kontorovich,lowryduda2017variants,oh2014limits, Shanks}.
This problem can be mapped to a problem whose solution is known, namely the Gauss circle problem. This pertains to finding the number $r_2(m,M)$ of solutions to the equation $k^2+l^2=M^2+m^2$. This number is determined by factoring $M^2+m^2$ into its prime factors~\cite{Shanks} and counting the number of primes, $p_i$, of the form $p_i\equiv1\,\mathrm{mod}\,4$~(this is described in detail in~\cite{bressoudwagon}; the dependence on~$M$ is a topic of current research~\cite{lowryduda2017variants,oh2014limits}).

This factorization procedure generates a sequence of primes that contains an element of inherent randomness. It is this property that captures the random distribution of the integral points on AdS$_2$~-- this is illustrated in Figs.~\ref{AdS2rat}.
\begin{figure}[t]\centering
\includegraphics[scale=0.4]{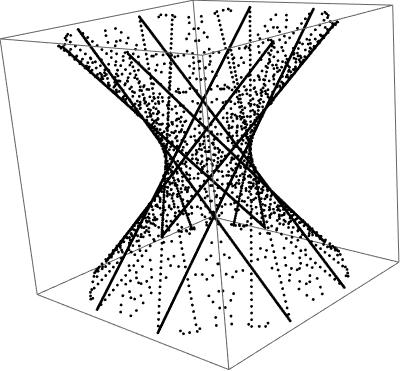}\qquad
\includegraphics[scale=0.4]{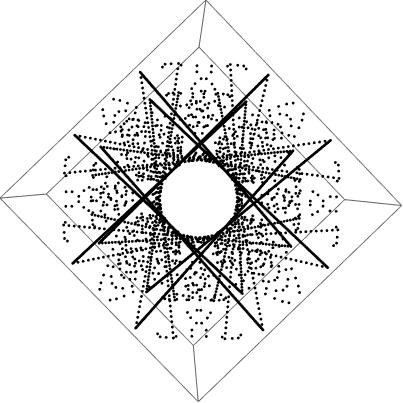}
\caption{Integral points on AdS$_2$.}\label{AdS2rat}
\end{figure}

Therefore, from these facts, the number of integral points of the hyperboloid, up to height~$m$, is given by the expression
\begin{equation*}%\label{solhyp}
\mathrm{Sol}(m)=4+2\sum_{j=1}^m r_2(j,M).
\end{equation*}
We plot this function -- in Fig.~\ref{GPpoints}, for $M=1$, when $m$ runs from $-200$ to 200 (due to the symmetry, $m\leftrightarrow -m$, we plot only the positive values of $m$).

\begin{figure}[t]\centering
\includegraphics[scale=0.2]{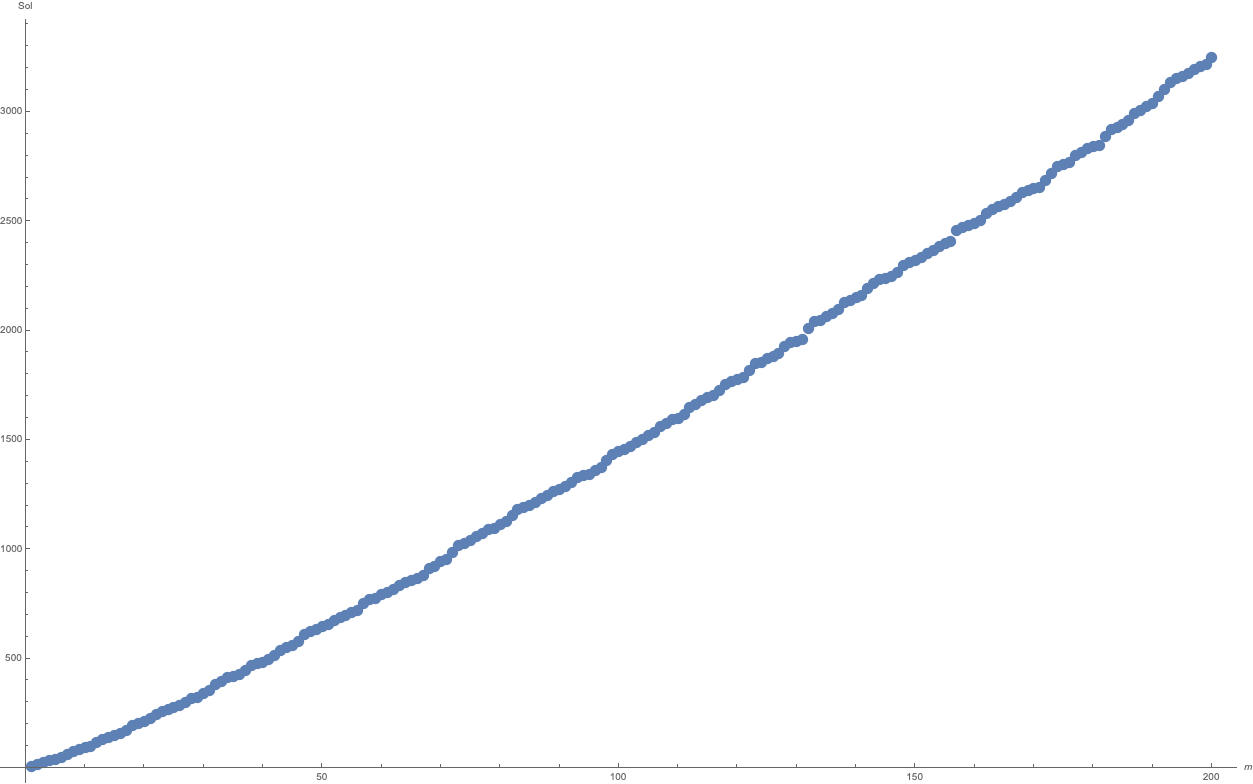}
\caption{The number of integral points, on AdS$_2$, as a function of the height, $m$, for $M=1$. Due to symmetry, $m\leftrightarrow -m$, we plot only the positive values of~$m$.}\label{GPpoints}
\end{figure}

It is, indeed, striking that the result is an almost straight line~\cite{lowryduda2017variants,oh2014limits}.

We shall now discuss how to actually construct these points, using the property that they belong to light-cone lines, which emerge from the rational points of the circle on the throat of AdS$_2$.

Using the ruling property of AdS$_2$,
\begin{gather*}%\label{coords}
k = \cos\phi - \mu\sin\phi,\qquad
l = \sin\phi + \mu\cos\phi, \qquad
m = \mu,
\end{gather*}
we may repackage these as follows
\begin{equation*}%\label{coords1}
x_0 + \mathrm{i}x_1 = k+\mathrm{i}l={\rm e}^{\mathrm{i}\phi}(1+\mathrm{i}\mu)={\rm e}^{\mathrm{i}\phi}(1+\mathrm{i}m)\Leftrightarrow {\rm e}^{\mathrm{i}\phi}=\frac{k+\mathrm{i}l}{1+\mathrm{i}m},
\end{equation*}
hence
\begin{equation}\label{rationalpointscircle}
\cos\phi = \frac{k+lm}{1+m^2} \qquad \mathrm{and}\qquad \sin\phi = \frac{l-mk}{1+m^2}.
\end{equation}
We remark that these are rational numbers~-- therefore they label rational points on the circle~\cite{tan}.

The light cone lines at $(k,l,m)$ are, therefore, parametrized by $\mu\in(-\infty,\infty)$, as
\begin{equation*}%\label{lcklm}
x_0 = \frac{k+lm}{1+m^2} -\mu\frac{l-mk}{1+m^2},\qquad
x_1 = \frac{l-mk}{1+m^2} +\mu\frac{k+lm}{1+m^2},\qquad
x_2 = \mu.
\end{equation*}
(When $\mu=x_2=m$, $x_0=k$ and $x_1 = l$.)
\begin{prop}On these specific light-cone lines there exist infinitely many integral points,
 when $\mu$, that labels the space-like direction~$x_2$, takes appropriate integer values.
\end{prop}
\begin{proof}We write
\begin{equation*}%\label{intpt1}
x_0(\mu)+\mathrm{i}x_1(\mu)={\rm e}^{\mathrm{i}\phi}(1+\mathrm{i}\mu),
\end{equation*}
where $\phi$ is defined by equation~(\ref{rationalpointscircle}).

We look for integer values of $\mu=n\in\mathbb{Z}$, such that $x_0(n)$ and $x_1(n)$ are, also, integers.
That is{\samepage
\begin{equation*}%\label{gaussianInt}
x_0(n)+\mathrm{i}x_1(n)=\frac{k+\mathrm{i}l}{1+\mathrm{i}m}(1+\mathrm{i}n)
\end{equation*}
should be a Gaussian integer and this can happen iff $(1+\mathrm{i}n)/(1+\mathrm{i}m)=a+\mathrm{i}b$ with $a,b\in\mathbb{Z}$.}

Therefore
\begin{equation*}%\label{gaussianint1}
1+\mathrm{i}n=(a-mb)+\mathrm{i}(am+b)\Leftrightarrow \begin{cases} 1 = a-mb, \\ n = am+b.\end{cases}
\end{equation*}

Thus on the light cone line passing through the point $(k,l,m)$ there are infinite integer points parametrized as
\begin{gather*}%\label{pointonAdS2}
x_0 = k + b(km-l), \qquad
x_1 = l + b(k+lm),\qquad
x_2 = n = m+ b\big(1+m^2\big).\tag*{\qed}
\end{gather*}\renewcommand{\qed}{}
\end{proof}

\begin{prop}Conversely, on any light cone line emanating from any rational point of the circle on the throat of the hyperboloid there is an infinite number of integer points.
\end{prop}
\begin{proof}Indeed, we have
\begin{equation*}%\label{gaussianint2}
{\rm e}^{\mathrm{i}\phi}\equiv\frac{a+\mathrm{i}b}{a-\mathrm{i}b}\Leftrightarrow x_0+\mathrm{i}x_1=\frac{a+\mathrm{i}b}{a-\mathrm{i}b}(1+\mathrm{i}n)
\end{equation*}
with $a,b\in\mathbb{Z}$. In order to obtain an integral point, for $\mu=n$, we must have
\begin{equation*}%\label{gaussianint3}
\frac{1+\mathrm{i}n}{a-\mathrm{i}b}=d+\mathrm{i}c
\end{equation*}
with $c,d\in\mathbb{Z}$.

We immediately deduce that
\begin{equation*}%\label{gaussianint4}
1 = ad-bc, \qquad n = ac + bd.
\end{equation*}
These expressions imply that, given the integers $a$ and $b$, it's possible to find the integers~$c$ and~$d$ and to express the coordinates $x_0$, $x_1$ and $x_2$ as
\begin{equation*}%\label{gaussianint5}
x_0 = ad+bc,\qquad x_1 = ac-bd, \qquad x_2 = ac+bd.
\end{equation*}
The Diophantine equation $1 = ad-bc$ is solved for $c$ and $d$, given two coprime integers $a$ and $b$, by the Euclidian algorithm~-- which seems to lead to a unique solution, implying that the point $(x_0, x_1, x_2)$ is unique.

However there's a subtlety! There are {\em infinitely many} solutions $(c,d)$, to the equation $ad-bc=1$! The reason is that, given any one solution $(c,d)$, the pair $(c+\kappa a, d+\kappa b)$, with $\kappa\in\mathbb{Z}$, is, also, a solution, as it can be checked by substitution.

Therefore there is a one-parameter family of points, labeled by the integer $\kappa$:
\begin{gather}\label{integerpointAdS2}
x_0 = ad+bc + 2\kappa a b,\qquad\!
x_1 = ac - bd + \kappa \big(a^2-b^2\big),\qquad\!
x_2 = ac+bd + \kappa \big(a^2+b^2\big).\!
\end{gather}
We remark, however, that the vector $\big(2ab,a^2-b^2,a^2+b^2\big)$ is light-like, with respect to the $(++-)$ metric: $(2ab)^2+\big(a^2-b^2\big)^2-\big(a^2+b^2\big)^2=0$. So equation~(\ref{integerpointAdS2})
describes a shift of the point $(ad+bc, ac-bd, ac+bd)$, along a light-like direction. Since the shift is linear in the ``affine parameter'', $\kappa$, it generates a light-like line, passing through the original point.

In this way we have established the dictionary between the rational points of the circle and the integral points of the hyperboloid.
\end{proof}

Now we proceed with the study of the discrete symmetries of the integral Lorentzian lattice of~$\mathscr{M}^{2,1}$, where the lattice of integral points on AdS$_2$ is embedded.
The lattice of integral points of~$\mathscr{M}^{2,1}$, with one space-like and two time-like dimensions, carries as isometry group the group of integral Lorentz boosts ${\rm SO}(2,1,\mathbb{Z})$, as well as integral Poincar\'e translations. The double cover of this infinite and discrete group is ${\rm SL}(2,\mathbb{Z})$, the modular group. This has been shown by Schild~\cite{Schild48, Schild49} in the 1940s. The group ${\rm SO}(2,1,\mathbb{Z})$ can be generated by reflections, as has been shown by Coxeter~\cite{Coxeter}, Vinberg~\cite{Vinberg_1967}. This work culminates in the famous book by Kac~\cite{Kac1990}, where he introduced the notion of hyperbolic, infinite-dimensional, Lie algebras. The characteristic property of such algebras is that the discrete Weyl group of their root space is an integral Lorentz group. Generalization from ${\rm SL}(2,\mathbb{Z})$ to other normed algebras has been studied in~\cite{Feingold:2008ih}.

The fundamental domain of ${\rm SO}(2,1,\mathbb{Z})$ is the minimum
 set of points of the integral lattice of~$\mathscr{M}^{2,1}$, which are
not related by any element of the group and from which,
 all the other points of the lattice can be generated by repeated action
of the elements of the group. It turns out that the fundamental region
is an infinite set of points which can be generated by repeated action of reflections in the following way:

Using the metric $h\equiv\mathrm{diag}(1,1,-1)$ on $\mathscr{M}^{2,1}$ the generating reflections, elements of ${\rm SO}(2,1,\mathbb{Z})$, are given by the matrices
\begin{gather*}%\label{genrefl}
R_1=\left(\begin{matrix}
-1 & & \\
 & 1 & \\
 & & 1
\end{matrix}\right), \qquad
R_2 = \left(\begin{matrix}
1 & & \\
 & 1 & \\
 & & -1
\end{matrix}\right), \qquad
R_3 = \left(\begin{matrix}
0 & 1 & \\
1 & 0 & \\
 & & 1
\end{matrix}\right),\\
R_4 = \left(\begin{matrix}
1 & -2 & -2\\
2 & -1 & -2\\
-2 & 2 & 3
\end{matrix}\right).
\end{gather*}
If $(k,l,m)$ are the coordinates of the integral lattice, the fundamental domain of ${\rm SO}(2,1,\mathbb{Z})$ can be defined by the conditions $m\geq k+l\geq 0$ and $k\geq l\geq 0$. This fundamental domain, restricted on AdS$_2^M[\mathbb{Z}]$, defines the corresponding fundamental domain of ${\rm SO}(2,1,\mathbb{Z})$, acting on AdS$_2^M[\mathbb{Z}]$.
This region of AdS$_2[\mathbb{Z}]$ lies in the positive octant of $\mathscr{M}^{2,1}$ and between the two planes, that define the conditions~-- cf.~Fig.~\ref{FundDomainAdS2ZN}. It is of {\em infinite} extent.

\begin{figure}[t!]\centering
\includegraphics[scale=0.5]{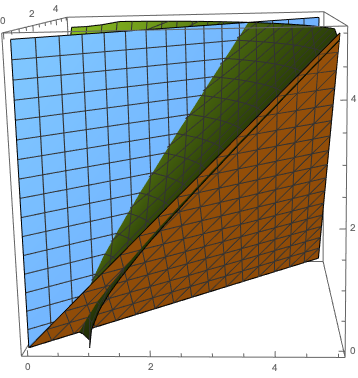}

\caption{The fundamental domain of ${\rm SO}(2,1,\mathbb{Z})$ on AdS$_2^M[\mathbb{Z}]$ is the dark green part of the hyperoboloid, in the positive octant, that lies between the two planes, $m\geq k+l\geq 0$ and $k\geq l\geq 0$.}\label{FundDomainAdS2ZN}
\end{figure}

\subsection[The IR cutoff and the toroidal compactification of AdS2]{The IR cutoff and the toroidal compactification of $\boldsymbol{{\rm AdS}_2}$}\label{UVIR}

Having introduced the lattice of integral points on AdS$_2$, which we consider as defining an UV cutoff, we proceed, now, to impose an infrared (IR) cutoff. The crucial reason for such a cutoff is that in order to study chaotic Hamiltonian dynamics on this spacetime~\cite{Axenides:2016nmf}, we have made use of the interpretation of AdS$_2$ as a phase space of single particles, due to the symplectic nature of the isometry ${\rm SL}(2,\mathbb{R})={\rm Sp}(2,\mathbb{R})$. The additional requirement of mixing (scrambling) imposes the condition of the compactness of the phase space and therefore the necessity of imposing of an infrared cutoff (for a detailed discussion of this point, cf.~\cite{ArnoldAvez}).

Having embedded the AdS$_2$ hyperboloid,
\begin{equation*}%\label{AdS2_M21}
x_0^2 + x_1^2 - x_2^2 = R_{{\rm AdS}_2}^2
\end{equation*}
in $\mathscr{M}^{2,1}$, the IR cutoff, $L$ is defined by periodically identifying all the spacetime points of $\mathscr{M}^{2,1}$, if the difference of their coordinates is an integral vector$\times L$:
\begin{equation*}%\label{equivclassIR}
x\sim y\Leftrightarrow x-y = (k,l,m)L,
\end{equation*}
where $k,l,m\in\mathbb{Z}$. In this way we have compactified $\mathscr{M}^{2,1}$ to the three-dimensional torus, of size~$L$, $\mathbb{T}^3(L)$.

More concretely, $\mathbb{T}^3(L)$ is the fundamental domain of the group of integral translations, $\mathbb{Z}\times\mathbb{Z}\times\mathbb{Z}$, acting on $\mathscr{M}^{2,1}$. To describe this geometric property by the algebraic operation, mod~$L$, that acts on the coordinates of $\mathscr{M}^{2,1}$, we are led to identify the fundamental domain with the positive octant of $\mathscr{M}^{2,1}$, i.e., $x_0,x_1,x_2\geq 0$

After this compactification, the spacetime geometry of AdS$_2$ becomes a foliation of the 3-torus, with leaves the images of AdS$_2$ under the operation mod~$L$. So the equation, whose solutions define the points of the compactified AdS$_2$, is
\begin{equation}
\label{AdS2comp}
x_0^2+x_1^2-x_2^2\equiv R_{{\rm AdS}_2}^2\,\mathrm{mod}\,L,
\end{equation}
where $(x_0,x_1,x_2)\in\mathbb{T}^3(L)$.

It is obvious, that inside the 3-torus, there is a part of the AdS$_2$ surface, which corresponds to solutions of equation~(\ref{AdS2comp}), without the mod $L$ operation. On the other hand, the infinite part of AdS$_2$, that lies outside the torus, is partitioned in infinitely many pieces, which belong to images of $\mathbb{T}^3(L)$ in $\mathscr{M}^{2,1}$.
 These pieces are brought inside the torus by the mod $L$ operation.

Now we choose the IR cutoff $L$ in units of $a$, so that $L=aN$, where $N$ is an integer, independent of $M$. It is constrained by $N>M$, since the cube should contain, at least, the throat of~AdS$_2$.

So the scaling limit entails taking $M\to\infty$, $N\to\infty$, but keeping $L$ fixed.

The periodic nature of the IR cutoff implies that we must take the images of all integral points of AdS$_2[\mathbb{Z}]$ under the mod $N$ operation, inside the cubic lattice of $N^3$ points.

The set of these images satisfy the equations
\begin{equation}\label{AdS2ZN}
k^2+l^2-m^2\equiv M^2\,\mathrm{mod}\,N.
\end{equation}
The set of points satisfying this condition will be called AdS$_2^M[\mathbb{Z}_N]$.

Our definition for AdS$_2[\mathbb{Z}_N]$ in our previous work was similar to the one given here. The only difference being that the r.h.s.\ of equation~(\ref{AdS2ZN}) was $1\,\mathrm{mod}\,N$, which was chosen for convenience, rather than for any intrinsic reason. We remark that the two definitions are consistent iff $M^2\equiv 1\,\mathrm{mod}\,N$.

The solutions of equation~(\ref{AdS2ZN}), when $M^2\equiv 1\,\mathrm{mod}\,N$, produce the AdS$_2[\mathbb{Z}_N]$ geometry introduced in our previous work.

\section[Continuum limit for large N]{Continuum limit for large $\boldsymbol{N}$}\label{contlim}
\subsection{Constraints on the double sequences of the UV/IR cutoffs}\label{M2N}

Having constructed the finite geometry, AdS$_2^M[\mathbb{Z}_N]$ and established its relation with AdS$_2[\mathbb{Z}_N]$, we shall discuss the meaning of the limit, $M,N\to\infty$. It is in this limit that we hope to recover the continuum AdS$_2$ geometry.

Such a limit can be defined using the topology of the ambient Minkowski spacetime $\mathscr{M}^{2,1}$.

Specifically, we use a reverse, two-step, process: Firstly, by removing the UV cutoff; next, by removing the IR cutoff. This is realized by choosing any sequence of pairs of integers, $(M_n,N_n)$, $n=1,2,3,\ldots$, such that, for any $n=1,2,3,\ldots$:
\begin{itemize}\itemsep=0pt
\item $N_n>M_n$,
\item $M_n^2\equiv 1\,\mathrm{mod}\,N_n$,
\item the limit of the ratio $N_n/M_n$ takes a finite value, $>1$ (as $n\to\infty$), which we can identify with $L/R_{{\rm AdS}_2}$.
\end{itemize}
Below we shall present the general solution to the equation $M^2\equiv 1\,\mathrm{mod}\,N$. Subsequently, we shall select those solutions that satisfy the other requirements.

The first step is to factor $N$ into (powers of) primes, $N=N_1\times N_2\times\cdots\times N_l=q_1^{k_1}q_2^{k_2}\cdots q_l^{k_l}$. Then the equation $M^2\equiv 1\,\mathrm{mod}\,N$, is equivalent to the system
\begin{equation}\label{M21N}
M_I^2\equiv 1\,\mathrm{mod}\,q_I^{k_I},
\end{equation}
where $I=1,2,\ldots,l$. The Chinese remainder theorem~\cite{bressoudwagon} then implies that all the solutions of equation~(\ref{M21N}) can be used to construct $M$, with $M=M_1m_1n_1 + \cdots + M_lm_ln_l$, where
$M_I\equiv M\,\mathrm{mod}\,N_I$, $m_I=N/N_I$, $n_I \equiv m_I^{-1}\,\mathrm{mod}\,N_I$.

When $q_I\neq 2$, the solutions are $M_I=1$ and $q_I^{n_I}-1$. When $q_I = 2$, there exist four solutions, $M_I=1,2^{n_I}-1,2^{n_I-1}\pm 1$.

Now we must choose sequences, $N_n$ and determine the corresponding $M_n$, satisfying the constraints listed above.

In the next two subsections we shall present nontrivial examples of sequences of pairs, $(M_n, N_n)$ satisfying the above constraints, whose limiting ratio, $\lim\limits_{n\to\infty} N_n/M_n$, is the ``golden'' or ``silver'' ratios.
The general question of determining sequences which have an arbitrary, but given, limiting ratio, is an interesting question, which is deferred to a future work.

\subsection{Removing the UV cutoff by the Fibonacci sequence}\label{fibon}
Although it is easy to demonstrate the existence of such sequences~-- for example, $N_n=2^n$ and $M_n=2^{n-1}\pm 1$, where $M_n^2\equiv 1\,\mathrm{mod}\,N_n$ and $N_n/M_n\to 2$, which implies that \mbox{$L/R_{{\rm AdS}_2}=2$}, in this section we focus on another particular class of sequences, based on the Fibonacci integers,~$f_n$~\cite{bressoudwagon}. This case is of particular interest, since, in our previous paper~\cite{Axenides:2016nmf}, where we studied fast scrambling, we found that, for geodesic observers, moving in AdS$_2[N]$, with evolution operator the Arnol'd cat map, the fast scrambling bound is saturated, when~$N$ is a Fibonacci integer.

The Fibonacci sequence, defined by
\begin{equation*}%\label{fibinacci_seq}
f_0=0, \qquad f_1=1, \qquad f_{n+1}=f_n+f_{n-1} ,
\end{equation*}
can be written in matrix form
\begin{equation*}%\label{matrixfib}
\left(\begin{matrix} f_n\\f_{n+1}\end{matrix}\right)=\underbrace{\left(\begin{matrix} 0 & 1 \\ 1 & 1\end{matrix}\right)}_{\sf A}\left(\begin{matrix} f_{n-1} \\ f_n\end{matrix}\right).
\end{equation*}
We remark that the famous Arnol'd cat map can be written as
\begin{equation*}%\label{ArnoldCM}
\left(\begin{matrix} 1 & 1 \\ 1 & 2\end{matrix}\right) = {\sf A}^2.
\end{equation*}
Since the matrix ${\sf A}$ doesn't depend on $n$, we can solve the recursion relation in closed form, by setting $f_n\equiv C \rho^n$ and find the equation, satisfied by $\rho$
\[
\rho^{n+1}=\rho^n+\rho^{n-1}\Leftrightarrow \rho^2-\rho-1=0\Leftrightarrow \rho\equiv \rho_\pm=\frac{1\pm\sqrt{5}}{2}.
\]
Therefore, we may express $f_n$ as a linear combination of $\rho_+^n$ and $\rho_-^n=(-)^n\rho_+^{-n}$:
\begin{equation*}%\label{solfib}
f_n=A_+\rho_+^n+A_-\rho_-^n\Leftrightarrow \begin{cases} f_0=A_+ + A_- = 0,\\ f_1 = A_+\rho_+ + A_-\rho_-=1,\end{cases}
\end{equation*}
whence we find that
\[
A_+ = -A_-=\frac{1}{\rho_+ -\rho_-}=\frac{1}{\sqrt{5}},
\]
therefore,
\begin{equation}\label{solfib1}
f_n=\frac{\rho_+^n-(-)^n\rho_+^{-n}}{\sqrt{5}}.
\end{equation}
It's quite fascinating that the l.h.s.\ of this expression is an integer!

The eigenvalue $\rho_+>1$ is known as the ``golden ratio'' (often denoted by $\phi$ in the literature) and it's straightforward to show that $f_{n+1}/f_n\to\rho_+$, as $n\to\infty$.

Furthermore, it can be shown, by induction, that the elements of ${\sf A}^n$ are, in fact, the Fibonacci numbers themselves, arranged as follows
\begin{equation*} %\label{fibonacciA}
 {\sf A}^n = \left(\begin{matrix} f_{n-1} & f_n\\ f_n & f_{n+1}\end{matrix}\right).
\end{equation*}
One reason this expression is useful is that it implies that $\det {\sf A}^n = (-)^n=f_{n-1}f_{n+1}-f_n^2$.

For $n=2l+1$, we remark that this relation takes the form $f_{2l+1}^2=1+f_{2l}f_{2l+2}$.

Now, since $f_{2l+1}$ and $f_{2l+2}$ are successive iterates, they're coprime, which implies, that $f_{2l+1}^2\equiv 1\,\mathrm{mod}\,f_{2l+2}$.

Therefore, the sequence of pairs,
$(M_l=f_{2l+1},N_l=f_{2l+2})$, where $l=1,2,3,\ldots$, satisfy all of the requirements and the
corresponding limiting ratio, $L/R_{{\rm AdS}_2}$, can be found analytically. It is, indeed, equal to $\rho_+=\big(1+\sqrt{5}\big)/2$, the golden ratio.

 In the next subsection we shall consider the so-called $k$-Fibonacci sequences, which will be important for obtaining other values for the ratio $L/R_{{\rm AdS}_2}$, as well as for removing the IR cutoff.

\subsection[Removing the IR cutoff using the generalized k-Fibonacci sequences]{Removing the IR cutoff using the generalized $\boldsymbol{k}$-Fibonacci sequences}\label{contfrac}

It's possible to generalize the Fibonacci sequence in the following way:
\begin{equation*}%\label{kfibrec}
g_{n+1}=kg_n+g_{n-1}
\end{equation*}
with $g_0=0$ and $g_1=1$ and $k$ an integer. This is known as the ``$k$-Fibonacci'' sequence~\cite{Horadam}.

We may solve for $g_n\equiv C\rho^n$; the characteristic equation for $\rho$, now, reads
\begin{equation*}%\label{kfibseq}
\rho^2-k\rho-1=0\Leftrightarrow\rho_\pm(k)=\frac{k\pm\sqrt{k^2+4}}{2}
\end{equation*}
and express $g_n$ as a linear combination of the $\rho_\pm$:{\samepage
\begin{equation*}%\label{kfibsol}
g_n = A_+\rho_+(k)^n + A_-\rho_-(k)^n=\frac{\rho_+(k)^n-(-)^n\rho_+(k)^{-n}}{\sqrt{k^2+4}}
\end{equation*}
that generalizes equation~(\ref{solfib1}).}

In matrix form
\begin{equation*}%\label{kfibinacci}
\left(\begin{matrix} g_n\\ g_{n+1}\end{matrix}\right)=\underbrace{\left(\begin{matrix} 0 & 1\\ 1 & k\end{matrix}\right)}_{{\sf A}(k)}\left(\begin{matrix} g_{n-1}\\g_n\end{matrix}\right).
\end{equation*}
Similarly as for the usual Fibonacci sequence, we may show, by induction, that
\begin{equation}\label{Akn}
{\sf A}(k)^n=\left(\begin{matrix} g_{n-1} & g_n \\ g_n & g_{n+1}\end{matrix}\right).
\end{equation}
We find that $\det {\sf A}(k)^n=(-)^n$, therefore that $g_{2l+1}^2\equiv 1\,\mathrm{mod}\,g_{2l+2}$; thus, $g_{2l+2}/g_{2l+1}\to L/R_{{\rm AdS}_2}\allowbreak =\rho_+(k)$, where the eigenvalue of ${\sf A}(k)$, $\rho_+(k)$, that's greater than 1, of course, depends on $k$. In this way it is possible to obtain infinitely many values of the ratio $L/R_{{\rm AdS}_2}$.
Furthermore, we have determined $L$, the IR cutoff, in terms of $R_{{\rm AdS}_2}$.

What is remarkable is that, using the
additional parameter, $k$, of the $k$-Fibonacci sequence, it is, now, possible to remove the IR cutoff, as well, since it is possible to send $L\to\infty$, as $k\to\infty$, keeping $R_{{\rm AdS}_2}$ fixed.

While $k$ remains finite, the periodic box cannot be removed and, in the continuum limit, $a\to 0$, we obtain infinitely many foldings of the AdS$_2$ surface inside the box due to the mod $L$ operation.

The Fibonacci sequence, taken mod $N$, is periodic, with period $T(N)$; this turns out to be a ``random'' function of $N$. The ``shortest'' periods, as has been shown by Falk and Dyson~\cite{falk_dyson}, occur when $N=F_l$, for any $l$. In that case, $T(F_l)=2l$.

We may, thus, ask the same question for the $k$-Fibonacci sequence, where the ratio of its successive elements, $g_{n+1}/g_n$ tend to the so-called ``$k$-silver ratio'',
\begin{equation*}%\label{silverratio}
\rho_+(k)=\frac{k+\sqrt{k^2+4}}{2}
\end{equation*}
(the ``silver ratio'' is $\rho_+(k=2)$).

From equation~(\ref{Akn}), taking mod $g_l$ on both sides, we find that, when $n=l$, the matrix becomes $\pm$(the identity matrix), so $T(g_l)=l$ or $2l$, respectively; thereby generalizing the Falk--Dyson result for the $k$-Fibonacci sequences.

\section{Conclusions and open issues}\label{concl}
In this work we have proposed a construction of the continuum AdS$_2$ radial and time near horizon geometry of extremal black holes from a finite and arithmetic geometry, AdS$_2[N]$, for every integer $N$. This entails the introduction of UV and IR cutoffs, respectively $a=R_{{\rm AdS}_2}/M$ and $L=a N$, where $L>R_{{\rm AdS}_2}$ is the size of the periodic box, that encloses the one-sheeted hyperboloid.

The periodic box and the UV cutoff deform the ${\rm PSL}(2,\mathbb{R})$ isometry of AdS$_2$ to the finite group, ${\rm PSL}_2[\mathbb{Z}_N]$, which is the mod $N$ reduction of ${\rm PSL}_2[\mathbb{Z}]$.

The elements of this finite group are discrete maps and describe the evolution operators of the avatars of infalling observers, with proper time the iteration time of the corresponding maps.

The notion of locality in gravity is expressed in terms of the diffeomorphism invariance of the gravitational action. This implies the absence of local observables and it is only in the case of well defined asymptotic behavior of the metric, either conformal or not, that globally defined observables do exist which can characterize the gravitational background. In the case of the AdS/CFT correspondence, the holographic dualities are restricted by the UV/IR correspondence and locality cannot hold simultaneously on the boundary as well as in the bulk.

On the other hand, the present efforts to understand the near horizon region, as well as the interior and the exterior of black holes, which are asymptotically anti-de Sitter, rely exclusively on the boundary CFT point of view. This approach, however, reaches its limit when attempting to resolve features, beyond the Planck scale, where no formalism for performing reliable calculations is, to date, available.

For these reasons our program for using the arithmetic of finite geometries has an intrinsic interest as an alternative way for reconstructing bulk spacetimes, as emerging in an appropriate scaling limit thereof. Among the main advantages are:
\begin{itemize}\itemsep=0pt
\item As shown in this paper this scaling limit is the correct one, in that the usual, continuum, AdS$_2$ geometry is recovered~-- this is a very important sanity check.
\item
The relation of finite geometries to quantum information theory and their representation as quantum circuits with measurable complexity~\cite{Jefferson:2017sdb,Preskill:1999he, Susskind:2018fmx,Susskind:2018pmk}.
\item
It, also, provides a framework for quantitatively studying the eigenstate thermalization hypothesis~\cite{SrednickiETH} and the fast scrambling bound~\cite{Axenides:2016nmf}.

Due to the modular arithmetic, an intrinsic number theoretic randomness appears in the geometry itself, as well as in the dynamics of wave packets with finite-dimensional Hilbert space~\cite{Athanasiu:1994fv}.
\end{itemize}

In the present work we established that the modular geometry AdS$_2^M[\mathbb{Z}_N]$ is a useful toy model that realizes many of the basic properties, for the near horizon geometries of extremal/near extremal black holes, in that it can be shown to lead to the definition of the correct continuum limit.

Along the way, we discussed interesting methods to localize and count the integral points of the AdS$_2$ continuous geometry and to characterize the points of AdS$_2^M[\mathbb{Z}_N]$ as equivalence classes of the AdS$_2$ integral points modulo the congruent modular group~$\Gamma[N]$. The continuum limit of the modular geometry AdS$_2^M[\mathbb{Z}_N]$ was constructed explicitly, using infinite sequences of UV/IR cutoffs $(M_n,N_n)$, $n=1,2,\ldots$, taken from the integer
sequences of the $k$-Fibonacci numbers.

The sequence of UV cutoffs$,N_n$ describes the dimension of the Hilbert space of states of single-particle probes and, in the case of $k$-Fibonacci sequence, $k=1,2,\ldots$, the dynamics of the corresponding cat maps saturates the scrambling time bound
with a Lyapunov exponent that grows logarithmically with~$k$.

Among the open issues of our approach we may mention:
\begin{itemize}\itemsep=0pt
\item Our approach to the continuum geometry consists in showing that the ratio $L/R_{{\rm AdS}_2}$ can take certain (though infinitely many) values; realizing the construction for arbitrary values of this ratio remains an open issue.
\item The distribution of the integral points of AdS$_2$ seems to have quite interesting properties~\cite{lowryduda2017variants,oh2014limits}.
\item The sequence of AdS$_2^M[\mathbb{Z}_N]$ modular geometries, for $N\to\infty$, can be studied in the framework of profinite integers and groups. The limit of this sequence belongs to
the set AdS$_2\big[\widehat{\mathbb{Z}}\big]$, where $\widehat{\mathbb{Z}}$ is the set of profinite integers.\footnote{We would like to thank one of the referees of this paper, who stressed the relevance of the profinite integers for the present construction.} The sequence of the UV/IR pairs can be lifted to the so-called profinite Fibonacci integers. Their limit can be, also, studied in the corresponding topology~\cite{lenstra2003profinite,lenstra2005profinite,lenstra2016profinite}.\footnote{We would like to thank Professor H.W.~Lenstra for correspondence on this point.}
\item The extension to modular discretizations of higher-dimensional AdS/CFT duals, using the corresponding arithmetic isometry groups.
\item Another possible direction to this end could the relation of the modular with the $p$-adic AdS$_2$ geometry~\cite{Gubser:2017pyx, Ma:2018mln} and referencs therein.
\item The extension to the BTZ black hole.
\item Describing de Sitter spacetimes~\cite{Witten:2001kn} using arithmetic geometry.
\item Many-body probe systems and the ensuing questions related to their entanglement and the time behaviour of their OTOC bulk quantum correlators.
\end{itemize}
These issues are technically feasible and physically interesting with available tools.

\subsection*{Acknowledegements} This work spanned many places and benefitted from discussions with many people. We would like to thank, in particular, Costas Bachas and John Iliopoulos at the LPTENS, Gia Dvali, Alex Kehagias, Boris Pioline, Kyriakos Papadodimas and Eliezer Rabinovici at CERN. We acknowledge the warm hospitality at Ecole Normale Sup\'erieure, Paris, the Theory Division at CERN and the Institute of Nuclear and Particle Physics of the NRCPS ``Demokritos''.

We would, also, wish to thank Professor H.W.~Lenstra for illuminating correspondence and the referees of our paper for the interest they showed in our submission and their detailed reports, that gave us the opportunity to sharpen our arguments and improve significantly the presentation.

\pdfbookmark[1]{References}{ref}
\LastPageEnding

\end{document}